\documentclass[pdflatex,sn-mathphys]{sn-jnl}

\usepackage{graphicx}%
\usepackage{multirow}%
\usepackage{amsmath,amssymb,amsfonts}%
\usepackage{amsthm}%
\usepackage{mathrsfs}%
\usepackage[title]{appendix}%
\usepackage{xcolor}%
\usepackage{textcomp}%
\usepackage{manyfoot}%
\usepackage{booktabs}%
\usepackage{algorithm}%
\usepackage{algorithmicx}%
\usepackage{algpseudocode}%
\usepackage{listings}%
\usepackage{a4}

\newtheorem{theorem}{Theorem}%
\newtheorem{definition}{Definition}%
\newtheorem{proposition}[theorem]{Proposition}%

\begin{document}

\title[The ergodic hypothesis: a typicality statement]{The ergodic hypothesis: a typicality statement}

\author*[1,2]{\fnm{Paula} \sur{Reichert}}\email{reichert@math.lmu.de}

\affil*[1]{\orgdiv{Mathematisches Institut}, \orgname{LMU München}, \orgaddress{\country{Germany}}}

\affil[2]{\orgdiv{Department for Humanities \& Arts}, \orgname{Technion, Israel Institute of Technology}, \orgaddress{\country{Israel}}}

\abstract{This paper analyzes the ergodic hypothesis in the context of Boltzmann's late work in statistical mechanics, where Boltzmann lays the foundations for what is today known as the typicality account. I argue that, based on the concepts of stationarity (of the measure) and typicality (of the equilibrium state), the ergodic hypothesis, as an idealization, is a consequence rather than an assumption of Boltzmann's approach. More precisely, it can be shown that every system with a stationary measure and an equilibrium state (be it a state of overwhelming phase space or time average) behaves essentially as if it were ergodic.  
I claim that Boltzmann was aware of this fact as it grounds both his notion of equilibrium, relating it to the thermodynamic notion of equilibrium, and his estimate of the fluctuation rates.
}


\maketitle

\section{Introduction}\label{sec1}

The ergodic hypothesis has been formulated by \cite{Boltzmann1871} and \cite{Maxwell} and has famously been discussed by \cite{Ehrenfest1911} in their influential encyclopedia article on statistical mechanics, where they provide an overview of and comment on Boltzmann's work in statistical physics. 

Ever since, the ergodic hypothesis has been debated controversially. This refers not only to the status of the ergodic hypothesis within Boltzmann's work (see, e.g., \cite{Brush1967}), but more generally to its applicability with respect to realistic systems (see, e.g., \cite{Earman}; \cite{Smale}) and its relevance for physics as such (see, e.g., \cite{Sklar}; \cite{Schwartz}; \cite{Bricmont}). 

Despite its debatable status, the concept of ergodicity has attracted a lot of attention. Today there even exists a proper branch of mathematics, so-called ergodic theory, with a plentitude of rigorous mathematical results (most notably, the results of \cite{Birkhoff}, \cite{vonNeumann}, and \cite{Khinchin}; see \cite{Petersen} for an overview).

Interestingly enough, though, Boltzmann himself never highlighted the ergodic hypothesis. Although he introduces it in his early work, he mentions it not even once in his two volumes on gas theory, which constitute his \textit{opus magnum} on statistical mechanics (cf. \cite{Boltzmann1896a}). Still, he seems to rely on ergodicity, at least as an idealization, also in his later work like, for instance, when he estimates the rate of fluctuations in the letter to Zermelo (cf. \cite{Boltzmann1896b}). 

This said, has ergodicity been a fundamental assumption of Boltzmann as the Ehrenfests suggest? If so, why didn't he make this more explicit? This seems the more surprising as he does emphasize the explanatory value of other concepts. For instance, he stresses the fact that equilibrium is a typical state, i.e., a state which  is realized by an overwhelming number of micro configurations, at several points throughout his work (see, e.g., \cite{Boltzmann1896a, Boltzmann1896b, Boltzmann1897}). 

In this paper, I argue that ergodicity, as an idealization, or essential ergodicity, in the strict sense (as defined in section 3.3 below), is a consequence rather than an assumption of Boltzmann's approach. Based on this, I claim that the ergodic hypothesis should be read as a typicality statement, in a way analogous to how Boltzmann taught us to read the H-theorem (see \cite{Boltzmann1896b, Boltzmann1897}). That is, just as a dynamical system of many particles doesn't approach equilibrium for {all}, but for {typical} initial conditions (given a low-entropy initial macrostate) and stays there not for {all}, but for {most} times, in the case of ergodicity, not {all}, but {typical} systems behave not {strictly}, but {essentially}, that is {qualitatively}, as if they were ergodic. 

To make this point precise, what can be shown is the following: \textit{On typical trajectories, the time and phase space averages of physical macrostates  coincide in good appoximation}. This property of the dynamics, which I call {`essential ergodicity'}, follows from the stationarity of the measure and the typicality of the equilibrium state {alone}.

\section{The ergodic hypothesis}

To discuss the ergodic hypothesis, we need to 
introduce the realm of Boltzmann's statistical mechanics: the theory of measure-preserving dynamical systems.

\subsection{Measure-preserving dynamical systems}

Let ($\Gamma, \mathcal{B}(\Gamma), T, \mu$) denote a Hamiltonian system. For $N$ particles, $\Gamma\cong \mathbb{R}^{6N}$ is called phase space. It is the space of all possible microstates $X$ of the system, where a point $X=(q,p)$ in $\Gamma$ represents the positions and momenta of all the particles: $(q,p)=({q}_1, ..., {q}_{3N}, {p}_1, ..., {p}_{3N})$. 

The Hamiltonian flow $T$ is a one-parameter flow $T^t(q,p)=(q, p)(t)$ on $\Gamma$ with $t$ representing time. It is connected to the Hamiltonian vector field $v_H$ as follows: $ v_H(T^t(q,p)) = {dT^t(q,p)}/{dt}.$ In other words, the flow lines are the integral curves along the Hamiltonian vector field, where the latter is specified by $v_H = (\partial H/\partial p, -\partial H/\partial q)$. This is the {physical} vector field of the system, generated by the Hamiltonian $H$, and the flow lines represent the possible trajectories of the system. Finally, $\mu$ refers to the Liouville measure,
\begin{equation} d\mu =\prod_{i=1}^{3N} dq_i dp_i,\end{equation} 
or to any other stationary measure derived thereof. 

Note that we call a measure $\mu$ \textit{stationary} (with respect to $T$) if and only if the flow $T$ is measure-preserving (with respect to $\mu$). Given a Hamiltonian system, it follows from {Liouville's theorem} that the Liouville measure is conserved under the Hamiltonian phase flow. That is, for every $A\in \mathcal{B}(\Gamma)$,
\begin{equation} \mu(T^{-t}A)=\mu(A).\end{equation} 
Since the Liouville measure is just the $6N$-dimensional Lebesgue measure, this says that phase space volume is conserved under time evolution.

If we introduce the notion of the time-evolved measure, $\mu_t (A) := \mu(T^{-t}A)$, we can reformulate the condition of stationarity as follows. A measure $\mu$ is \textit{stationary} if and only if, for every $A\in \mathcal{B}(\Gamma)$,
\begin{equation}\mu_t(A) = \mu(A).\end{equation}
According to this equation, the measure itself is invariant under time translation, which is the main reason for physicists to accept it as \textit{the} measure grounding a statistical analysis in physics (see, e.g., \cite{Ehrenfest1911}, \cite{Gibbons}, \cite{Duerretal}). 
In practice, we are not concerned with the Liouville measure \textit{per se}, but with appropriate stationary measures derived thereof.\footnote{Consider, for instance, an isolated system. Within that system, total energy $E$ is conserved. Hence, trajectories are restricted to the constant-energy hypersurface $\Gamma_E = \{(q,p)\in \Gamma| H(q,p)=E\}$, from which it follows that the microcanonical measure
 \[d\mu_E = \prod_{i=1}^{3N} dq_i dp_i\ \delta(H(q,p)-E)\]
is the appropriate stationary measure of the dynamics in that case.} 
 
\subsection{Variants of the ergodic hypothesis}

Within the framework of Hamiltonian systems or, more generally, measure-preserving dynamical systems, we can analyze Boltzmann's ergodic hypothesis. 

Let again $(\Gamma, \mathcal{B}(\Gamma), T, \mu)$ be a measure-preserving dynamical system and $A\in \mathcal{B}(\Gamma)$. Let, in what follows, $\mu(\Gamma)=1$.\footnote{Throughout this paper, we deal with systems where $\Gamma$ is finite and, hence, $\mu$ is normalizable. In that case, we can set $\mu(\Gamma)=1$ without loss of generality. The hard case of infinite phase spaces has to be discussed elsewhere (see \cite{Goldstein2016} and \cite{Lazarovici2020} for a first discussion).} We call 
\begin{equation} \mu(A) = \int_\Gamma \chi_A(x) d\mu(x)\end{equation}
the {`phase space average'} of $A$ with $\chi_A$ being the characteristic function which is 1 if $x\in A$ and 0 otherwise. Further we call
\begin{equation}\hat{A}(x)=  \lim_{\mathcal{T}\to \infty} \frac{1}{\mathcal{T}} \int_{0}^{\mathcal{T}} \chi_{A}(T^tx)dt\end{equation} 
the {`time average'} of $A$ for some $x\in\Gamma$. Here it has been proven by \cite{Birkhoff} that the infinite-time limit exists pointwise almost everywhere on $\Gamma$ and the limit function $\hat{A}(x)$ is integrable. 

A dynamical system is called \textit{ergodic} if and only if, for all $A\in \mathcal{B}(\Gamma)$ and almost all $x\in\Gamma$ (i.e. for all $x$ except a measure-zero set), the time and phase averages coincide:
\begin{equation} \mu(A) =\hat{A}(x).\end{equation}

\noindent
In other words, a system is called ergodic if and only if, for almost all solutions, the fraction of time the system spends in a certain region in phase space (in the limit $t\to \infty$!) is precisely \textit{equal} to the phase space average of that region.

Historically, the ergodic hypothesis has been formulated differently. In its original version due to \cite{Boltzmann1871} (cited by \cite{Ehrenfest1911}), it refers to the assertion that a trajectory literally has to \textit{go through every point} in phase space (more precisely, in the constant-energy hypersurface). But this would imply that there is only \textit{one} solution with all possible microstates belonging to one and the same solution. This has been proven impossible by \cite{Rosenthal} and \cite{Plancherel}. 

In a weaker formulation, the so-called `quasi-ergodic hypothesis' demands that a trajectory has to \textit{come arbitrarily close to every point} in phase space (see \cite{Ehrenfest1911}). 
Later, the results of \cite{Birkhoff} and \cite{vonNeumann} established the precise conditions under which equality of the time and phase space average is obtained.\footnote{\cite{Birkhoff} gives a definition of ergodicity in terms of invariant sets (where a set $A\in \mathcal{B}(\Gamma)$ is called invariant if and only if $T^{-1}A=A$). If, for all sets $A\in \mathcal{B}(\Gamma)$ with $T^{-1}A=A$,
\[ \mu(A)=0 \hspace{1cm} \mathrm{or} \hspace{1cm}\mu(A) =1,\]
then the system is called `ergodic'.
Thus a system is called `ergodic' if and only if all invariant sets are of full or zero measure. In other words, there exist no two (or more) disjoint invariant sets of non-zero measure. The two definitions of ergodicity relate to one another via Birkhoff's theorem.} 

For realistic physical systems, this equality of the time and phase space average -- that is, ergodicity -- turned out to be extremely hard to prove, if it could be proven at all. To draw on the most important result: it took almost 50 years and joined efforts to extend Sinai's proof for the model of 1 billiard ball on a 2-dimensional table (cf. \cite{Sinai}) to the generalized model of $N\ge 2$ hard spheres in a container with periodic boundary conditions (i.e. a torus) of dimension $d\ge 2$; see \cite{Simanyi}. 

At this point, the question arises: What if we were not interested in the exact coincidence of the time and phase space average in the first place? What if all we need is an approximate equality of the time and phase space average on typical trajectories? The point I want to make is the following: {Boltzmann, being concerned with the analysis of realistic physical systems, need not be and presumably was not interested in ergodicity in the strict sense}. According to \cite{Ehrenfest1911}, Boltzmann used ergodicity to estimate the fraction of time a system spends in a certain macrostate. To obtain such an estimate, however, it suffices to establish a result qualitatively comparable to ergodicity: \textit{an {almost equality} of the time and phase space average of physical macrostates on {typical} trajectories.} This is precisely where the notion of essential ergodicity comes into play.

\section{Essential ergodicity}

We need one last ingredient to grasp the notion of essential ergodicity and that is the notion of typicality of macro- and microstates. We will then find that, given a stationary measure and a typical macrostate, that is, an equilibrium state in Boltzmann's sense, a typical system behaves essentially as if it were ergodic. 

\subsection{Typicality and Boltzmann's notion of equilibrium}

Given a measure on the space of possible states of the system -- like a volume measure on phase space -- this is naturally a measure of probability or typicality.\footnote{There is a little caveat to this statement. While it is definitely true whenever phase space is finite and the measure is normalizable, one has to be careful with infinite phase spaces and non-normalizable measures. For problems related to the latter, see \cite{Schiffrin} or \cite{Goldstein2016}. The distinction between the notions of probability and typicality has been drawn and discussed elsewhere (see, e.g., \cite{Goldstein2012}, \cite{Lazarovici2015} or \cite{Wilhelm}).} 
Let again $\mu$ denote the volume measure on $\Gamma$. We call a measurable set $A\subset \Gamma$ `typical'  (with respect to $\Gamma$) if and only if
\begin{equation} \mu(A) = 1 -\varepsilon \end{equation}
for $0<\varepsilon <<1$. This definition of `typical sets' directly entails a definition of `typical points' (cf. \cite{Wilhelm}).
We say that a point $x$ is `typical' (with respect to $\Gamma$) if and only if $x\in A$ and $A$ is typical with respect to $\Gamma$.

In Boltzmann's statistical mechanics, we are concerned with `points' (microstates) and `sets' (macro-regions). Macro-regions are regions of phase space corresponding to physical macrostates of the system. More precisely, every microstate $X$, represented by a point $({q},{p})$ on $\Gamma$, belongs to respectively determines a certain macrostate $M(X)$, represented by an entire region $\Gamma_M\subset \Gamma$ -- the set of all microstates realizing that particular macrostate. 
While a microstate comprises the exact positions and velocities of all the particles, $X=({q}_1, ..., {q}_N, {p}_1, ..., {p}_N$), a macrostate $M(X)$ is specified by the macroscopic, thermodynamic variables of the system, like volume $V$, temperature $T$, and so on. By definition, any two macrostates $M_i$ and $M_j$ are macroscopically distinct, hence there are only finitely many macrostates $M_i$, and all macrostates together provide a partition of phase space into disjoint `macro regions' $\Gamma_{M_i}$  with $\Gamma =\bigcup_{i=1}^n \Gamma_{M_i}$. Here it is a consequence of the large number of particles that every macrostate $M(X)$ is realized by a huge number of microstates $X$ and, hence, the precise way of partitioning doesn't matter.

In this set-up, Boltzmann defined `equilibrium' precisely as the \textit{typical} macrostate of the system. 
\vspace{0.3cm}

\begin{definition}[Boltzmann equilibrium] Let $(\Gamma, \mathcal{B}(\Gamma), T, \mu)$ be a dynamical system. Let $\Gamma$ be partitioned into finitely many disjoint, measurable subsets $\Gamma_{M_i}, i=1, .., n$ by some (set of) physical macrovariable(s) $M_i$, i.e., $\Gamma =\bigcup_{i=1}^n \Gamma_{M_i}$.  Then a set $\Gamma_{Eq}\in\{\Gamma_{M_1}, ..., \Gamma_{M_n}\}$ with phase space average 
\begin{equation}\label{be}\mu(\Gamma_{Eq})=1-\varepsilon\end{equation} 
where $\varepsilon\in\mathbb{R}$, $ 0<\varepsilon<<1$, is called the `equilibrium set' or {`equilibrium region'}. The corresponding macrostate $M_{Eq}$ is called the {`Boltzmann equilibrium'} of the system.
\end{definition}
\vspace{0.3cm}

Be aware that this definition is grounded on a particular, physical macro partition of phase space. In other words, it is not an arbitrary value of $\varepsilon$ which, when given, determines an equilibrium state -- such a definition would be meaningless from the point of physics. Instead, it is a partition determined by the physical macrovariables of the theory, which is given, and it is with respect to that partition that a region of overwhelming phase space measure, if it exists, defines an equilibrium state in Boltzmann's sense (and by the way determines the value of $\varepsilon$). 

At this point, it has been Boltzmann's crucial insight that, for a realistic physical system of $N \approx 10^{24}$ particles (where, for a medium-sized object, we take Avogadro's constant) and a partition into {macroscopically distinct} states, there always exists a region of overwhelming phase space measure (see, e.g., \cite{Boltzmann1896b}).\footnote{\cite{Lanford} proves the existence of a region of overwhelming phase space measure for a large class of realistic physical systems.} This follows essentially from the vast gap between micro and macro description of the system and the fact that, for a large number of particles, small differences at the macroscopic level translate into huge differences in the corresponding phase space volumes. 

To obtain an idea of the numbers, consider a gas in a medium-sized box. For that model, \cite{Penrose1989, Penrose2004} estimates the volume of all non-equilibrium regions together as compared to the equilibrium region to be:
\begin{equation}  \frac{\mu(\bigcup_{i=1}^n \Gamma_{M_i}\setminus \Gamma_{{Eq}})}{\mu(\Gamma_{{Eq}})}= \frac{\mu(\Omega\setminus \Gamma_{{Eq}})}{\mu(\Gamma_{{Eq}})}\approx 1: 10^N\end{equation}
with $N\approx 10^{24}$. This implies, with $\mu(\Gamma_{Eq})\approx \mu(\Omega)$, that $\varepsilon$ is of the order $1:10^N\approx 1:10^{10^{24}}=\frac{1}{10^{1000000000000000000000000}}.$


Both Boltzmann's realization that equilibrium is a typical state and his understanding that any two distinct macrostates relate to macro-regions that differ vastly in size provided the grounds for his explanation of irreversible behaviour (cf. \cite{Boltzmann1896a, Boltzmann1896b}; see \cite{Lebowitz1993a, Lebowitz1993b}, \cite{Bricmont}, \cite{Goldstein2001}) for further elaboration of this point). In the following, however, we are only concerned with ergodicity and, related to that, a system's long-time behaviour. 

\subsection{Precise bounds on the time and phase space average of the equilibrium state}

In what follows, we give precise bounds on the time average of the equilibrium state. Therefore, consider a dynamical system with a stationary measure $\mu$ and an equilibrium state $\Gamma_{eq}$ in the sense of Boltzmann. That is, $\mu(\Gamma_{Eq})=1-\varepsilon$. 

To be able to formulate the bound on the time average and, later, the notion of `essential ergodicity', we have to distinguish between a `good' set G and a `bad' set B of points $x \in \Gamma$. Let, in what follows, $B$ be the `bad' set of points for which the time average of equilibrium $\hat{\Gamma}_{Eq}(x)$ is smaller than $1-k\varepsilon$ (with $1\le k\le1/\varepsilon$). All points in this set determine trajectories which spend a fraction of less than $1-k\varepsilon$ of their time in equilibrium. Let further $G$ be the `good' set of points with a time average $\hat{\Gamma}_{Eq}(x)$ of at least $1-k\varepsilon$. All points in this set determine trajectories that spend a fraction of at least $1-k\varepsilon$ of their time in equilibrium. To be precise,
\begin{equation}B:=\{x\in\Gamma|\hat{\Gamma}_{Eq}(x)<1-k\varepsilon\}, \hspace{0.5cm}G:=\{x\in\Gamma|\hat{\Gamma}_{Eq}(x)\ge1-k\varepsilon\}.\end{equation} 

While, for a realistic physical system, ergodicity is hard to prove -- if it can be proven at all --, essential ergodicity is not. In fact, it follows almost directly from the stationarity of the measure and the typicality of the equilibrium state.
To be precise, with respect to the two sets $B$ and $G$ the following can been shown. For all $\varepsilon, k\in\mathbb{R}$ with $0<\varepsilon<<1$ and $1\le k\le 1/\varepsilon$: 
\begin{equation}\label{1}\mu(B) <1/k, \hspace{1cm} \mu(G)> 1- 1/k.\end{equation}
The proof can be found in the appendix (see also \cite{Reichert}). An essential ingredient entering the proof is the pointwise existence and integrability of the time average (cf. \cite{Birkhoff}). Hence, in the case of non-ergodic systems, the time average of equilibrium need not attain a fix value on (almost all of) $\Gamma$ -- in fact, it may have different values on different trajectories --, but still it \textit{exists} (pointwise almost everywhere) and this suffices to estimate the size of the set of trajectories with a time average smaller (or larger) than a particular value. 

To grasp the full meaning of Eq.~\ref{1}, consider a physically relevant value of $k$. Recall that, for a medium-sized macroscopic object, $\varepsilon$ is tiny: $\varepsilon\approx 10^{-N}$ with $N\approx 10^{24}$. In that case, one can choose $k$ within the given bounds ($1\le k\le 1/\varepsilon$) large enough for $\mu(B)$ to be close to zero and $\mu(G)$ to be close to one. Consider, for example,
\begin{equation} k=1/\sqrt{\varepsilon}.\end{equation}
In that case, we distinguish between the `good' set $G$ of trajectories which spend at least $1-\sqrt{\varepsilon}$ of their time in equilibrium and the `bad' set $B$ of trajectories which spend less than $1-\sqrt{\varepsilon}$ of their time in equilibrium. And we obtain:
\begin{equation} \mu(B)<\sqrt{\varepsilon}, \hspace{1cm} \mu(G)>1-\sqrt{\varepsilon}.\end{equation}
Given the value of $\varepsilon$ from above, $\varepsilon\approx 10^{-10^{24}}$, it follows that $\sqrt{\varepsilon}\approx 10^{-10^{23}}$. Consequently, the equilibrium region is of measure $\mu(\Gamma_{Eq})\approx1-10^{-10^{24}}$ and the measures of the sets $B$ and $G$ are 
\begin{equation}\label{est} \mu(B)<10^{-10^{23}}, \hspace{1cm} \mu(G)>1-10^{-10^{23}}.\end{equation}
Note that $B$ is now the set of trajectories which spend less than $1-10^{-10^{23}}$ and $G$ the set of trajectories which spend at least $1-10^{-10^{23}}$ (!) of their time in equilibrium. We thus find that trajectories which spend {almost all of their time in equilibrium are {typical} whereas trajectories which spend less than almost all of their time in equilibrium are {atypical}}!

The converse statement has be proven as well (\cite{Frigg2015a, Frigg2015b}; see the appendix for a different proof; cf. \cite{Reichert}). It says that if there exists a region $\Gamma_{Eq'}\subset \Gamma$ in which by far most trajectories spend by far most of their time, then this region has very large phase space measure. To be precise, if there exists a region $G'$ with $\mu(G')=1-\delta$ such that $\forall x\in G'$: $\hat{\Gamma}_{Eq'}(x)\ge 1-\varepsilon'$, then the following holds:

\begin{equation}\label{2}\mu(\Gamma_{Eq'})\ge (1-\varepsilon')(1-\delta).\end{equation} 
Here we are again interested in those cases where $\delta$ and $\varepsilon'$ are very small, $0<\delta<<1$ and $0<\varepsilon'<<1$ (while the result holds for other values of $\delta$ and $\varepsilon'$ as well). 

This converse result tells us that, if there exists a state in which a typical trajectory spends by far most of its time, then this state is of overwhelming phase space measure. 

Why is this converse statement interesting? It doesn't start from Boltzmann's notion of equilibrium. Instead, it starts from a thermodynamic or thermodynamic-like notion of equilibrium. 

According to a standard thermodynamics textbook (like, e.g., \cite{Callen} or \cite{Reiss}), a thermodynamic equilibrium is a state in which a system, once it is in that state, stays for all times. In what follows, we give a definition which relaxes that standard definition a little bit in that it allows for rare fluctuations out of equilibrium and for some atypical trajectories (all $x \notin G'$) that don't behave thermodynamic-like.\footnote{\cite{Lavis2005, Lavis2008} would call this a `thermodynamic-like equilibrium' to draw the distinction between this notion and the standard textbook definition.}
\vspace{0.3cm}

\begin{definition}[Thermodynamic equilibrium] 
 Let $(\Gamma, \mathcal{B}(\Gamma), T, \mu)$ be a dynamical system. Let $\Gamma$ be partitioned into finitely many disjoint, measurable subsets $\Gamma_{M_i} (i=1,...,n)$ by some (set of) physical macrovariable(s) $M_i$, i.e., $\Gamma =\bigcup_{i=1}^n \Gamma_{M_i}$. Let $G'\subset \Gamma$ with $\mu(G')=1-\delta$ and $0<\delta<<1$. Let $0<\varepsilon' << 1$. A set $\Gamma_{Eq'}\in\{\Gamma_{M_1}, ..., \Gamma_{M_n}\}$ (connected to a macrostate $M_{Eq'}$) with time average 
\begin{equation}\label{te}\hat{\Gamma}_{Eq'}(x) \ge 1-\varepsilon'\end{equation}
for all $x\in G'$ is called a `thermodynamic equilibrium'.
\end{definition}
\vspace{0.3cm}


To summarize, we obtain that, for every dynamical system with a stationary measure and a state of overwhelming phase space measure, almost all trajectories spend almost all of their time in that state, and the other way round, given a state in which almost all trajectories spend almost all of their time, that state is of overwhelming phase space measure.
\textit{Hence, an equilibrium state in Boltzmann's sense is a thermodynamic equilibrium and the other way round!}\footnote{Based on the apparently missing connection between the time and the phase space average of equilibrium, Frigg and Werndl assert that Boltzmann's account of thermodynamic behaviour, which has later become known as the `typicality account', is simply `mysterious' \cite[p. 918]{Frigg2012b}. In follow-up papers (cf. \cite{Frigg2015a, Frigg2015b}) they even claim that the typicality account doesn't relate to thermodynamics \textit{at all} because it doesn't draw the connection between Boltzmann's definition of equilibrium (in terms of the phase space average) and the thermodynamic definition of equilibrium (in terms of the time average). 
Here essential ergodicity counters the critique and closes the explanatory gap as it connects the time and phase space averages of the equilibrium state in a mathematically precise way.}

The only two assumptions which enter the proofs of the above assertions are:
\begin{itemize}
\item[a)] that the measure is stationary (resp. the dynamics is measure-preserving), i.e., $\mu_t(A)=\mu(A)$ for all $A\in \mathcal{B}(\Gamma)$ and 
\vspace{0.2cm}

\item[b)] that there is a macrostate of overwhelming phase space measure, i.e., a Boltzmann equilibrium $\Gamma_{Eq}$  with $\mu(\Gamma_{Eq})=1-\varepsilon$, \end{itemize}
or, for the reverse direction, a) and
\begin{itemize}
\item[c)] that there is a state in which typical trajectories spend by far most of their time, i.e., a thermodynamic equilibrium $\Gamma_{Eq'}$ with $\hat{\Gamma}_{Eq'}\ge 1-\varepsilon'$.
\end{itemize}
Ergodicity doesn't enter the proofs, nor do we get ergodicity out of it. However, we get something similar to ergodicity, 
what we call `{essential ergodicity}'. 

\subsection{Essential ergodicity}

While, for an ergodic system, the time and phase space averages \textit{exactly coincide for {all} but a measure-zero set of solutions}, for an essentially ergodic system, the time and phase space averages \textit{almost coincide on {typical} solutions}. To be precise, the following definition applies.
\vspace{0.3cm}

\begin{definition}[Essential ergodicity] Let  $(\Gamma, \mathcal{B}(\Gamma), T, \mu)$ be a dynamical system. Let $\Gamma$ be partitioned into finitely many disjoint, measurable subsets $\Gamma_{M_i} (i=1,...,n)$ by some (set of) physical macrovariable(s) $M_i$, i.e., $\Gamma =\bigcup_{i=1}^n \Gamma_{M_i}$. Let $0<\varepsilon<<1$. A system is called `essentially ergodic' if and only if 
\begin{equation}\label{ess}|\hat{\Gamma}_{M_i}(x)-\mu(\Gamma_{M_i})|\le \varepsilon\end{equation}
$\forall i=1, ..., n$ and $\forall x\in G$ with $\mu(G)\ge 1-\delta$, $0< \delta<<1$.
\end{definition}
\vspace{0.3cm}

 For a measure-preserving system with an equilibrium state (be it a Boltzmann or a thermodynamic equilibrium), Equations \ref{ess} follow in a straightforward way from the two definitions of equilibrium given in Eq.~\ref{be}  and Eq.~\ref{te} and the corresponding results on the time and phase space average, Eq.~\ref{est} and Eq.~\ref{2}, respectively.
 More precisely, the following holds.

\vspace{0.3cm}

\begin{theorem}[FAPP ergodic hypothesis] Let  $(\Gamma, \mathcal{B}(\Gamma), T, \mu)$ be a measure-preserving dynamical system. Let there be an equilibrium state $M_{Eq}$ (a Boltzmann or thermodynamic equilibrium) with corresponding equilibrium region $\Gamma_{Eq}\subset \Gamma$. 

Then the system is essentially ergodic. In particular, there exists an $\varepsilon\in \mathbb{R}$ with $0<\varepsilon<<1$ such that
\begin{equation}\label{erg}|\hat{\Gamma}_{Eq}(x)-\mu(\Gamma_{Eq})|\le \varepsilon\end{equation}
$\forall x\in G$ with $\mu(G)\ge 1-\delta$, $0< \delta<<1$.
\end{theorem}
\vspace{0.3cm}

\begin{proof} We only prove Equation \ref{erg}. From that, Equations \ref{ess} follow directly. 

Let $0<\delta', \varepsilon', \varepsilon''<<1$. For the first direction of proof, consider a thermodynamic equilibrium, i.e., $\hat{\Gamma}_{Eq}(x) \ge 1-\varepsilon'$ for all $x\in G'$ with $\mu(G')=1-\delta'$. It follows from Eq.~\ref{2} that $\mu(\Gamma_{Eq})\ge(1-\varepsilon')(1-\delta')$ and, hence, 
\begin{equation}|\hat{\Gamma}_{Eq}(x)-\mu(\Gamma_{Eq})|\le \varepsilon' + \delta' - \varepsilon' \delta'.\end{equation}
Now set $G=G'$, $\delta= \delta'$ and $\varepsilon=\varepsilon' + \delta' - \varepsilon' \delta'$.

For the other direction, consider a Boltzmann equilibrium, i.e., $\mu({\Gamma}_{Eq}) = 1-\varepsilon''$. It follows from Eq.~\ref{1} that $\mu(G'')>1-\sqrt{\varepsilon''}$
with $G''=\{x\in\Gamma|\hat{\Gamma}_{Eq}(x)\ge1-\sqrt{\varepsilon''}\}$. Hence, for all $x\in G''$,
\begin{equation}|\hat{\Gamma}_{Eq}(x)-\mu(\Gamma_{Eq})|\le \varepsilon''.\end{equation}
Now set $G=G''$, $\delta = \sqrt{\varepsilon''}$ and $\varepsilon=\varepsilon''$.
\end{proof}

Bear in mind that, in this theorem, the order of $\varepsilon$ is the order of the incredibly tiny proportion of phase space that is occupied by the system's non-equilibrium macrostates. This means that for all practical purposes (FAPP) the time and phase space averages can be taken to be equal. In other words, the system behaves essentially as if it were ergodic.     

\subsection{Scope and limits of (essential) ergodicity}

Although the notion of essential ergodicity is weaker than the notion of ergodicity, it predicts qualitatively the same long-time behaviour. In particular, it tells us that a typical trajectory spends by far most of its time in equilibrium, where equilibrium is defined in Boltzmann's way in terms of the phase space average, and it makes this notion of `by far most' mathematically precise.\footnote{Goldstein makes a similar point when he asserts that, even without ergodicity, the value of any thermodynamic variable is constant `to all intents and purposes' \cite[p. 46]{Goldstein2001}.} 
This justifies, in a rigorous way, Boltzmann's assumption of ergodicity as an idealization or FAPP truth in analyzing the system's long-time behaviour (as done, e.g., in his estimate of the fluctuation rate \cite{Boltzmann1896b}). 
In other words, based on Boltzmann's account, the ergodic hypothesis is well-justified. It is a good working hypothesis for those time scales on which it begins to matter that trajectories wind around all of phase space.

Let us, at this point, use the above result on essential ergodicity to estimate the rate of fluctuations out of equilibrium. Recall that, according to Eq.~\ref{est}, typical trajectories spend at least $1-10^{-10^{23}}$ of their time in equilibrium, when equilibrium is of measure $\mu(\Gamma_{Eq})=1-10^{-10^{24}}$ (which is a reasonable value for a medium-sized object). In other words, they spend a fraction of less than $10^{-10^{23}}$ of their time {out of} equilibrium, that is, in a fluctuation. If we assume that fluctuations happen randomly, in accordance with a trajectory wandering around phase space erratically, we obtain the following estimate for typical trajectories: a fluctuation of $1$ second occurs about every $10^{10^{23}}$ seconds. But this means that a typical medium-sized system spends trillions of years in equilibrium as compared to one second in non-equilibrium, a time larger than the age of the universe!\footnote{This agrees with the time estimate Boltzmann presents in his letter to Zermelo \cite[p. 577]{Boltzmann1896b}.}

So far we argued that essential ergodicity substantiates Boltzmanns assertions about the long-time behaviour of macroscopic systems. What about the short-time behaviour? In physics and philosophy, several attempts have been made to use ergodicity in some way or the other to explain a system's evolution from non-equilibrium to equilibrium (see \cite{Vranas} or \cite{Frigg2011, Frigg2012a}; for earlier attempts as well as a thorough critique, see \cite{Bricmont} and the references therein). 

In this paper, I argue that ergodicity -- just like epsilon-ergodicity, essential ergodicity, or any other notion involving an infinite-time limit -- \textit{does not} and \textit{cannot} tell us anything about the approach to equilbrium, which is a behaviour within \textit{short times}. 
This is simply due to the fact that the notion of ergodicity (or any notion akin to that) involves an \textit{infinite-time limit}. Because of that limit, ergodicity can, at best, tell us something about the system's \textit{long-time} behaviour where `long-time' refers to time scales comparable to the recurrence times, where it begins to matter that the system's trajectory winds around all of phase space. For those short time scales on which the system evolves from non-equilibrium to equilibrium, ergodicity (or any notion akin to that) doesn't play any role. In fact, for a realistic gas, the equilibration time scale (i.e. the time scale of a system's approach to equilibrium) is fractions of a second as compared to trillions of years for the recurrence time! 

Boltzmann's explanation of the irreversible approach to equilibrium is a genuine typicality result (see the discussion and references at the end of section 3.1) -- ergodicity doesn't add to nor take anything from that.

At this point, a quote of the mathematician Schwartz fits well.\footnote{This quote was one of the first quotes (and essays) that were given to me by Detlef D\"urr, to whom this memorial volume is dedicated. 
} Schwartz writes with respect to Birkhoff's ergodic theorem and the widely-spread conception that ergodicity might help to explain thermodynamic behaviour \cite[pp. 23--24]{Schwartz}:

 \begin{quote}The intellectual attractiveness of a mathematical argument, as well as the considerable mental labor involved in following it, makes mathematics a powerful tool of intellectual prestidigitation -- a glittering deception in which some are entrapped, and some, alas, entrappers. Thus, for instance, the delicious ingenuity of the Birkhoff ergodic theorem has created the general impression that it must play a central role in the foundations of statistical mechanics. [...] The Birkhoff theorem in fact does us the service of establishing its own inability to be more than a questionably relevant superstructure upon [the] hypothesis [of typicality]. \end{quote}

\section{Conclusion}

Based on typicality and stationarity as the two basic concepts of Boltzmann's approach, it follows that ergodicity, as an idealization, or essential ergodicity, in the strict sense, is a consequence rather than an assumption of Boltzmann's account. 

I believe that Boltzmann was aware of this fact. In my opinion, he simply didn't highlight the precise mathematical connection between the concepts of typicality, stationarity, and essential ergodicity because it was absolutely clear to him that, given a state of overwhelming phase space volume and a stationary measure, by far most trajectories would stay in that state by far most of their time -- just like by far most trajectories starting from non-equilibrium would move into equilibrium very quickly. He didn't need a mathematical theorem to make this more precise.

Let me now end this paper with a variation of the both picturesque and paradigmatic example of Tim Maudlin, about typicality incidents occurring in the Sahara desert.\footnote{Known to the author from private conversation. The original version is 
about a person's approach from non-equilibrium (here: an oasis) to equilibrium (here: the remainder of the desert), where it is the atypical initial condition, the special fact of `being in an oasis' \textit{in the very beginning}, which is in need of explanation. The fact that a person, walking around in an unspecific and maybe even random way, walks out of the oasis into the desert is merely typical (we call it \textit{typical within atypicality}; see \cite{DuerrTeufel} for this phrasing). According to \cite{Goldstein2001}, it is the explanation of the atypical initial condition which constitutes the hard part of any explanation of thermodynamic irreversibility.} In what follows, I will adapt this example to the case of essential ergodicity.

A person wandering through the Sahara is typically surrounded by sand by far most of her time. In other words, she is typically hardly ever in an oasis. This fact is independent of the exact form of her `wandering about', if she changes direction often, or not, if she moves fast, or not, and so on. Even if she doesn't move at all, she is typically surrounded by sand (in that case, for all times). In other words, independent of the dynamics, the long-time average of `being surrounded by sand' is close to one on typical trajectories. This follows solely from the fact that all oases together constitute a vanishing small part of the Sahara desert and remain to do so throughout all times.\\
\newpage

\Large \noindent \textbf{Appendix}
\normalsize
\vspace{0.5cm}

\noindent In what follows, I prove a theorem on the time average of the Boltzmann equilibrium.\\

\begin{theorem}[Time average of $\Gamma_{Eq}$]
Let ($\Gamma, \mathcal{B}(\Gamma), \mu$) be a probability space and let $T$ be a measure-preserving transformation. Let $\varepsilon, k\in\mathbb{R}$ with $0<\varepsilon<<1$ and $1\le k\le 1/\varepsilon$. Let $\Gamma_{Eq}\subset\Gamma$ be an equilibrium region with $\mu(\Gamma_{Eq})=1-\varepsilon$. 
Let $B$ be the set of points for which the time average of equilibrium is smaller than $1-k\varepsilon$, $B=\{x\in\Gamma|\hat{\Gamma}_{Eq}(x)<1-k\varepsilon\}$. 
It follows that $B$ is of measure
\begin{equation} \mu(B) <1/k. \end{equation}
Let further $G=\{x\in\Gamma|\hat{\Gamma}_{Eq}(x)\ge1-k\varepsilon\}$ be the set of points for which the time average of equilibrium is larger than or equal to $1-k\varepsilon$. Then 
\begin{equation} \mu(G)> 1-1/k.\end{equation}
\end{theorem}


\begin{proof}[Proof (Theorem 4.1)]
The transformation $T$ is measure-preserving, that is, for any set $A\in \mathcal{B}(\Gamma)$ and $\forall t$: $\mu(A)=\mu(T^{-t}A)$. Hence, in particular, $\mu(\Gamma_{Eq})=\mu(T^{-t}\Gamma_{Eq})$ where $\Gamma_{eq}$ refers to the equilibrium state, i.e., $\mu(\Gamma_{Eq})=1-\varepsilon$. It follows that $\mu(T^{-t}\Gamma_{Eq})=1-\varepsilon$, as well, and thus:
\begin{eqnarray}\label{e1}1-\varepsilon &=&\mu(\Gamma_{Eq})= \mu(T^{-t}\Gamma_{Eq}) = \int_\Gamma \chi_{T^{-t}\Gamma_{Eq}}(x)d\mu(x)=\int_\Gamma \chi_{\Gamma_{Eq}}(T^tx)d\mu(x)\nonumber\\
&=& \lim_{\mathcal{T}\to \infty} \frac{1}{\mathcal{T}} \int_{0}^{\mathcal{T}} dt \int_\Gamma \chi_{\Gamma_{Eq}}(T^tx)d\mu(x).\end{eqnarray}
Here the last equation follows from the fact that the integrand is a constant.

At this point, we make use of the pointwise ergodic theorem of \cite{Birkhoff}\footnote{For a thorough presentation of Birkhoff's theorem and its proof, see \cite{Petersen}.} which says that, for any measure-preserving transformation $T$ 
and for any $\mu$-integrable function $f$, i.e. $f\in L^1(\mu)$, the limit 
\begin{equation} \hat{f}=\lim_{\mathcal{T}\to \infty} \frac{1}{\mathcal{T}} \int_{0}^{\mathcal{T}} f(T^tx)dt\end{equation}
exists for almost every $x\in\Gamma$ and the (almost everywhere defined) limit function $\hat{f}$ is integrable, i.e., $\hat{f}\in L^1(\mu)$. 

Let us apply Birkhoff's theorem to the above equation. The characteristic function $\chi_{\Gamma_{Eq}}$ is $\mu$-integrable. Hence, for almost all $x\in \Gamma$, $\lim_{\mathcal{T}\to \infty} \frac{1}{\mathcal{T}} \int_{0}^{\mathcal{T}} \chi_{\Gamma_{Eq}}(T^tx)dt$ exists and is $\mu$-integrable. In other words, for almost every single trajectory the time average exists. By dominated convergence, we can thus change the order of integration and pull the limit into the $\mu$-integral. Let $\Gamma^*\subset\Gamma$ denote the set of points for which the time average exists, with $\mu(\Gamma^*)=\mu(\Gamma)$. Then Eq.~\ref{e1} becomes
\begin{equation}\label{e2} 1-\varepsilon=\int_{\Gamma^*} d\mu(x) \bigg[\lim_{\mathcal{T}\to \infty} \frac{1}{\mathcal{T}} \int_{0}^{\mathcal{T}} \chi_{\Gamma_{Eq}}(T^tx)dt\bigg].\end{equation}

Let us analyze the general case.\footnote{It is interesting to demonstrate how this equation is fulfilled in the two `extreme' cases of possible dynamics: first, the ergodic case, which says that the trajectory is dense in phase space. Second, the case of $T^t$ being the identity which implies that every trajectory is merely one point. All other cases lie in between. 

The first way to fulfill Eq.~\ref{e2} is that the time average $\hat{\Gamma}_{Eq}(x)$ is a constant (almost everywhere). In that case, it must hold that $\hat{\Gamma}_{Eq}(x)=1-\varepsilon$. The set of all points $x\in \Gamma$ for which the limit exists (and is constantly $1-\varepsilon$), defines an invariant set, $T^{-1}\Gamma^*=\Gamma^*,$ with measure $\mu(\Gamma^*)=1$. This is the ergodic case. 

The second way to fulfill Eq.~\ref{e2} is that there exists an invariant set $A$ (i.e. $T^{-1}A=A$) with $\mu(A)=\varepsilon$ such that $\forall x\in A$: $ \hat{\Gamma}_{Eq}(x)=0$ and $\forall x \notin A:  \hat{\Gamma}_{Eq}(x)=1$ (again, up to a set of measure zero). Then also $\Gamma^*\backslash A$ is an invariant set and $\mu(\Gamma^*\backslash A)=1-\varepsilon$. This reflects the case of $T^t$ being the identity, $T^tx=x$, and $\Gamma\backslash A=\Gamma_{Eq}$.} Let again 
$G=\{x\in \Gamma| \hat{\Gamma}_{Eq}(x)\ge 1-k\varepsilon\}$
and $B=\{x\in \Gamma| \hat{\Gamma}_{Eq}(x)< 1-k\varepsilon\}.$
It is clear that this defines a decomposition of $\Gamma^*$ into disjoint sets, $\Gamma^*=G\cup B$, with $\mu(B)=\mu(\Gamma\backslash G)= 1-\mu(G)$, and where $G$ and $B$ are invariant sets.  
Hence, Eq.~\ref{e2} can be rewritten as 
\begin{equation}\label{e3} 1-\varepsilon
= \int_{G} \hat{\Gamma}_{Eq}(x) d\mu(x) + \int_{B} \hat{\Gamma}_{Eq}(x) d\mu(x).
\end{equation}
Let now the {`mean time average'} of $G$ be defined as 
\begin{equation}\bar{\Gamma}_{Eq}(G)= \frac{1}{\mu(G)}\int_G \hat{\Gamma}_{Eq}(x)d\mu(x),\end{equation}
where $\hat{\Gamma}_{Eq}(x)$ exists and is integrable for all $x\in G$. The mean time average determines the mean fraction of time the trajectories starting in $G$ spend in the set $\Gamma_{Eq}$. Analogously, let $\bar{\Gamma}_{Eq}(B)$ denote the mean time average of $B$. 
With this definition, Eq.~\ref{e3} can be rewritten as
\begin{equation}\label{e4}1-\varepsilon= \bar{\Gamma}_{Eq}(G)\mu(G) + \bar{\Gamma}_{Eq}(B)\mu(B). \end{equation}
We want to solve this for $\mu(B)$. Recall that $\mu(G)=1-\mu(\Gamma\backslash G)=1-\mu(B)$. 
Moreover, since $\lim_{T\to \infty} \frac{1}{T} \int_{0}^{T} \chi_{\Gamma_{Eq}}(T^tx)dt\le1$, it is $\bar{\Gamma}_{Eq}(G)\le1$. 
On the other hand, it follows from the definition of the mean time average that $\bar{\Gamma}_{Eq}(B)<1-k\varepsilon$ (since $\hat{\Gamma}_{Eq}(x)<1-k\varepsilon$ for all $x\in B$). Hence, since $k\ge 1$, it is $\bar{\Gamma}_{Eq}(B)<1-\varepsilon$.

Now in order for the right hand side of Eq.~\ref{e4} to add up to $1- \varepsilon$, the measure of $B$ needs to be small. This is due to the fact that $\mu(B)$ comes with a factor $\bar{\Gamma}_{Eq}(B)<1-\varepsilon$ which can only be encountered by a factor $\bar{\Gamma}_{Eq}(G)\ge 1-\varepsilon$ in front of $\mu(G)$. However, since $\bar{\Gamma}_{Eq}(G)$ is bounded from above by One, $\bar{\Gamma}_{Eq}(G)\le 1$, the first summand can outweigh the second only if $\mu(G)$ is large enough (respectively, $\mu(B)$ small enough). At most, $\bar{\Gamma}_{Eq}(G)=1$. In that case, $\mu(G)$ attains its minimum and $\mu(B)$ its maximum (where $\mu(B) = 1- \mu(G)$). 
Since we want to determine an upper bound of $\mu(B)$, we set $\bar{\Gamma}_{Eq}(G)=1$ (a condition we will relax later). Let, in addition, $\Theta:=\mu(B)$. Then equation Eq.~\ref{e4} can be rewritten as 
\begin{equation} 1-\varepsilon = (1-\Theta) + \bar{\Gamma}_{Eq}(B)\Theta \end{equation} 
With $\bar{\Gamma}_{Eq}(B)<1 -k\epsilon$, it follows that
\begin{equation} \hspace{0.5cm}\Theta=\frac{\varepsilon}{1-  \bar{\Gamma}_{Eq}(B)} < \frac{\varepsilon}{1- (1-k\varepsilon)} =1/k. \end{equation}
If we now no longer restrict the mean time average of $G$ to be one, this inequality becomes even more pronounced. That way we obtain an upper bound on $\mu(B)$:
\begin{equation}\mu(B)< 1/k.\end{equation} 
From this it follows directly that 
\begin{equation}\mu(G)=\mu(\Gamma\backslash B) = 1-\mu(B))> 1- 1/k.\end{equation}
This proves the assertion.
\end{proof}

\noindent In what follows, I give the proof of the converse statement saying that a state in which typical solutions stay by far most of their time is a state of by far largest phase space volume.\\

\begin{proposition}
Let the setting be as in the above theorem. Let $0<\delta<<1$ and $0<\varepsilon<<1$. Let now $\Gamma_{Eq'}\subset \Gamma$ and $G'\subset \Gamma$ with $\mu(G')=1-\delta$ such that $\forall x \in G'$: $\hat{\Gamma}_{Eq'}(x)\ge1-\varepsilon$. 
Then 
\begin{equation} \mu(\Gamma_{Eq'}) \ge (1-\varepsilon)(1-\delta).\end{equation}
\end{proposition}

\begin{proof}
When one applies Eq.~\ref{e1}, Eq.~\ref{e2} and Eq.~\ref{e3} to the set $\Gamma_{Eq'}\subset\Gamma$, one gets 
\begin{eqnarray} \mu(\Gamma_{Eq'})&=& \int_{\Gamma^*} d\mu(x) \bigg[\lim_{\mathcal{T}\to \infty} \frac{1}{\mathcal{T}} \int_{0}^{\mathcal{T}} \chi_{\Gamma_{Eq'}}(T^tx)dt\bigg]\nonumber\\
&=&\int_{G'} d\mu(x) \hat{\Gamma}_{Eq'}(x)+ \int_{B'} d\mu(x)\hat{\Gamma}_{Eq'}(x),\end{eqnarray}
where the first equality holds due to Birkhoff's theorem and the second uses the definition of the time average and $B'=\Gamma^*\backslash G'$. From $ \int_{B'} d\mu(x)\hat{\Gamma}_{Eq'}(x)\ge 0$ and the assumptions it follows that 
\begin{equation} \mu(\Gamma_{Eq'})\ge\int_{G'} d\mu(x) \hat{\Gamma}_{Eq'}(x)\ge (1-\varepsilon)(1-\delta). \end{equation}
This proves the assertion.
\end{proof}

\bibliography{sn-bibliography}


\begin{thebibliography}{46}
\ifx \bisbn   \undefined \def \bisbn  #1{ISBN #1}\fi
\ifx \binits  \undefined \def \binits#1{#1}\fi
\ifx \bauthor  \undefined \def \bauthor#1{#1}\fi
\ifx \batitle  \undefined \def \batitle#1{#1}\fi
\ifx \bjtitle  \undefined \def \bjtitle#1{#1}\fi
\ifx \bvolume  \undefined \def \bvolume#1{\textbf{#1}}\fi
\ifx \byear  \undefined \def \byear#1{#1}\fi
\ifx \bissue  \undefined \def \bissue#1{#1}\fi
\ifx \bfpage  \undefined \def \bfpage#1{#1}\fi
\ifx \blpage  \undefined \def \blpage #1{#1}\fi
\ifx \burl  \undefined \def \burl#1{\textsf{#1}}\fi
\ifx \doiurl  \undefined \def \doiurl#1{\url{https://doi.org/#1}}\fi
\ifx \betal  \undefined \def \betal{\textit{et al.}}\fi
\ifx \binstitute  \undefined \def \binstitute#1{#1}\fi
\ifx \binstitutionaled  \undefined \def \binstitutionaled#1{#1}\fi
\ifx \bctitle  \undefined \def \bctitle#1{#1}\fi
\ifx \beditor  \undefined \def \beditor#1{#1}\fi
\ifx \bpublisher  \undefined \def \bpublisher#1{#1}\fi
\ifx \bbtitle  \undefined \def \bbtitle#1{#1}\fi
\ifx \bedition  \undefined \def \bedition#1{#1}\fi
\ifx \bseriesno  \undefined \def \bseriesno#1{#1}\fi
\ifx \blocation  \undefined \def \blocation#1{#1}\fi
\ifx \bsertitle  \undefined \def \bsertitle#1{#1}\fi
\ifx \bsnm \undefined \def \bsnm#1{#1}\fi
\ifx \bsuffix \undefined \def \bsuffix#1{#1}\fi
\ifx \bparticle \undefined \def \bparticle#1{#1}\fi
\ifx \barticle \undefined \def \barticle#1{#1}\fi
\bibcommenthead
\ifx \bconfdate \undefined \def \bconfdate #1{#1}\fi
\ifx \botherref \undefined \def \botherref #1{#1}\fi
\ifx \url \undefined \def \url#1{\textsf{#1}}\fi
\ifx \bchapter \undefined \def \bchapter#1{#1}\fi
\ifx \bbook \undefined \def \bbook#1{#1}\fi
\ifx \bcomment \undefined \def \bcomment#1{#1}\fi
\ifx \oauthor \undefined \def \oauthor#1{#1}\fi
\ifx \citeauthoryear \undefined \def \citeauthoryear#1{#1}\fi
\ifx \endbibitem  \undefined \def \endbibitem {}\fi
\ifx \bconflocation  \undefined \def \bconflocation#1{#1}\fi
\ifx \arxivurl  \undefined \def \arxivurl#1{\textsf{#1}}\fi
\csname PreBibitemsHook\endcsname

\bibitem[\protect\citeauthoryear{Boltzmann}{1871}]{Boltzmann1871}
\begin{barticle}
\bauthor{\bsnm{Boltzmann}, \binits{L.}}:
\batitle{Einige allgemeine {S}{ä}tze {ü}ber {W}{ä}rmegleichgewicht}.
\bjtitle{Wiener Berichte}
\bvolume{63}(\bissue{2}),
\bfpage{679}
(\byear{1871})
\end{barticle}
\endbibitem

\bibitem[\protect\citeauthoryear{Maxwell}{1879}]{Maxwell}
\begin{barticle}
\bauthor{\bsnm{Maxwell}, \binits{J.C.}}:
\batitle{On {B}oltzmann's theorem on the average distribution of energy in a
  system of material points 1878}.
\bjtitle{Transactions of the Cambridge Philosophical Society}
\bvolume{12},
\bfpage{547}
(\byear{1879})
\end{barticle}
\endbibitem

\bibitem[\protect\citeauthoryear{Ehrenfest and
  Ehrenfest-Afanassjewa}{1911}]{Ehrenfest1911}
\begin{bchapter}
\bauthor{\bsnm{Ehrenfest}, \binits{P.}},
\bauthor{\bsnm{Ehrenfest-Afanassjewa}, \binits{T.}}:
\bctitle{Begriffliche {G}rundlagen der statistischen {A}uffassung in der
  {M}echanik}.
In: \bbtitle{Encyklop{ä}die der {m}athematischen Wissenschaften Mit Einschluss
  {i}hrer Anwendungen},
vol. \bseriesno{4},
pp. \bfpage{4}--\blpage{90}.
\bpublisher{Teubner},
\blocation{Leipzig}
(\byear{1911})
\end{bchapter}
\endbibitem

\bibitem[\protect\citeauthoryear{Brush}{1967}]{Brush1967}
\begin{barticle}
\bauthor{\bsnm{Brush}, \binits{S.G.}}:
\batitle{Foundations of statistical mechanics 1845--1915}.
\bjtitle{Archive for History of Exact Sciences}
\bvolume{4}(\bissue{3}),
\bfpage{145}--\blpage{83}
(\byear{1967})
\end{barticle}
\endbibitem

\bibitem[\protect\citeauthoryear{Earman and R{é}dei}{1996}]{Earman}
\begin{barticle}
\bauthor{\bsnm{Earman}, \binits{J.}},
\bauthor{\bsnm{R{é}dei}, \binits{M.}}:
\batitle{Why ergodic theory does not explain the success of equilibrium
  statistical mechanics}.
\bjtitle{The British Journal for the Philosophy of Science}
\bvolume{47}(\bissue{1}),
\bfpage{63}--\blpage{78}
(\byear{1996})
\end{barticle}
\endbibitem

\bibitem[\protect\citeauthoryear{Smale}{1979}]{Smale}
\begin{bchapter}
\bauthor{\bsnm{Smale}, \binits{S.}}:
\bctitle{On the problem of reviving the ergodic hypothesis of {B}oltzmann and
  {B}irkhoff}.
In: \beditor{\bsnm{Cucker}, \binits{F.}},
\beditor{\bsnm{Wong}, \binits{R.}} (eds.)
\bbtitle{The Collected Papers of Stephen Smale. Volume 2},
pp. \bfpage{823}--\blpage{30}.
\bpublisher{Singapore University Press},
\blocation{Singapore}
(\byear{1979})
\end{bchapter}
\endbibitem

\bibitem[\protect\citeauthoryear{Sklar}{1973}]{Sklar}
\begin{barticle}
\bauthor{\bsnm{Sklar}, \binits{L.}}:
\batitle{Statistical explanation and ergodic theory}.
\bjtitle{Philosophy of Science}
\bvolume{40}(\bissue{2}),
\bfpage{194}--\blpage{212}
(\byear{1973})
\end{barticle}
\endbibitem

\bibitem[\protect\citeauthoryear{Schwartz}{1992}]{Schwartz}
\begin{bchapter}
\bauthor{\bsnm{Schwartz}, \binits{J.}}:
\bctitle{The pernicious influence of mathematics on science}.
In: \beditor{\bsnm{Kac}, \binits{M.}},
\beditor{\bsnm{Rota}, \binits{G.C.}},
\beditor{\bsnm{Schwartz}, \binits{J.}} (eds.)
\bbtitle{Discrete Thoughts: Essays on Mathematics, Science and Philosophy},
pp. \bfpage{19}--\blpage{25}.
\bpublisher{Birkh{ä}user},
\blocation{Boston}
(\byear{1992})
\end{bchapter}
\endbibitem

\bibitem[\protect\citeauthoryear{Bricmont}{1995}]{Bricmont}
\begin{barticle}
\bauthor{\bsnm{Bricmont}, \binits{J.}}:
\batitle{Science of chaos or chaos in science?}
\bjtitle{Physicalia Magazine}
\bvolume{17},
\bfpage{159}--\blpage{208}
(\byear{1995})
\end{barticle}
\endbibitem

\bibitem[\protect\citeauthoryear{Birkhoff}{1931}]{Birkhoff}
\begin{barticle}
\bauthor{\bsnm{Birkhoff}, \binits{G.D.}}:
\batitle{Proof of the ergodic theorem}.
\bjtitle{Proceedings of the National Academy of Sciences of the United States
  of America}
\bvolume{17},
\bfpage{656}--\blpage{60}
(\byear{1931})
\end{barticle}
\endbibitem

\bibitem[\protect\citeauthoryear{von Neumann}{1932}]{vonNeumann}
\begin{barticle}
\bauthor{\bsnm{Neumann}, \binits{J.}}:
\batitle{Proof of the quasi-ergodic hypothesis}.
\bjtitle{Proceedings of the National Academy of Sciences}
\bvolume{18}(\bissue{1}),
\bfpage{70}--\blpage{82}
(\byear{1932})
\end{barticle}
\endbibitem

\bibitem[\protect\citeauthoryear{Khinchin}{1949}]{Khinchin}
\begin{bbook}
\bauthor{\bsnm{Khinchin}, \binits{A.I.}}:
\bbtitle{Mathematical Foundations of Statistical Mechanics}.
\bpublisher{Dover Publications},
\blocation{Richmond}
(\byear{1949})
\end{bbook}
\endbibitem

\bibitem[\protect\citeauthoryear{Petersen}{1983}]{Petersen}
\begin{bbook}
\bauthor{\bsnm{Petersen}, \binits{K.}}:
\bbtitle{Ergodic Theory. Cambridge Studies in Advanced Mathematics 2}.
\bpublisher{Cambridge University Press},
\blocation{Cambridge}
(\byear{1983})
\end{bbook}
\endbibitem

\bibitem[\protect\citeauthoryear{Boltzmann}{1896a}]{Boltzmann1896a}
\begin{bbook}
\bauthor{\bsnm{Boltzmann}, \binits{L.}}:
\bbtitle{Vorlesungen {ü}ber Gastheorie}.
\bpublisher{Verlag v. J. A. Barth (Nabu Public Domain Reprint)},
\blocation{Leipzig}
(\byear{1896})
\end{bbook}
\endbibitem

\bibitem[\protect\citeauthoryear{Boltzmann}{1896b}]{Boltzmann1896b}
\begin{barticle}
\bauthor{\bsnm{Boltzmann}, \binits{L.}}:
\batitle{Entgegnung auf die w{ä}rmetheoretischen {B}etrachtungen des {H}rn.
  {E}. {Z}ermelo}.
\bjtitle{Wiedemann's Annalen}
\bvolume{57},
\bfpage{773}--\blpage{84}
(\byear{1896})
\end{barticle}
\endbibitem

\bibitem[\protect\citeauthoryear{Boltzmann}{1897}]{Boltzmann1897}
\begin{barticle}
\bauthor{\bsnm{Boltzmann}, \binits{L.}}:
\batitle{Zu {H}rn. {Z}ermelos {A}bhandlung ``{Ü}ber die mechanische
  {E}rkl{ä}rung irreversibler {V}org{ä}nge'}.
\bjtitle{Wiedemann's Annalen}
\bvolume{60},
\bfpage{392}--\blpage{98}
(\byear{1897})
\end{barticle}
\endbibitem

\bibitem[\protect\citeauthoryear{Gibbons et~al.}{1987}]{Gibbons}
\begin{barticle}
\bauthor{\bsnm{Gibbons}, \binits{G.W.}},
\bauthor{\bsnm{Hawking}, \binits{S.W.}},
\bauthor{\bsnm{Stuart}, \binits{J.M.}}:
\batitle{A natural measure on the set of all universes}.
\bjtitle{Nuclear Physics B}
\bvolume{281},
\bfpage{736}--\blpage{51}
(\byear{1987})
\end{barticle}
\endbibitem

\bibitem[\protect\citeauthoryear{D{ü}rr et~al.}{2017}]{Duerretal}
\begin{bbook}
\bauthor{\bsnm{D{ü}rr}, \binits{D.}},
\bauthor{\bsnm{Fr{ö}mel}, \binits{A.}},
\bauthor{\bsnm{Kolb}, \binits{M.}}:
\bbtitle{Einf{ü}hrung in die {W}ahrscheinlichkeitstheorie {a}ls {T}heorie der
  {T}ypizit{ä}t}.
\bpublisher{Springer},
\blocation{Berlin}
(\byear{2017})
\end{bbook}
\endbibitem

\bibitem[\protect\citeauthoryear{Goldstein et~al.}{2016}]{Goldstein2016}
\begin{botherref}
\oauthor{\bsnm{Goldstein}, \binits{S.}},
\oauthor{\bsnm{Tumulka}, \binits{R.}},
\oauthor{\bsnm{Zangh{ì}}, \binits{N.}}:
Is the hypothesis about a low entropy initial state of the universe necessary
  for explaining the arrow of time?
Physical Review D
\textbf{94}(023520)
(2016)
\end{botherref}
\endbibitem

\bibitem[\protect\citeauthoryear{Lazarovici and
  Reichert}{2020}]{Lazarovici2020}
\begin{bchapter}
\bauthor{\bsnm{Lazarovici}, \binits{D.}},
\bauthor{\bsnm{Reichert}, \binits{P.}}:
\bctitle{Arrow(s) of time without a past hypothesis}.
In: \bbtitle{Statistical Mechanics and Scientific Explanation. Determinism,
  Indeterminism and Laws of Nature},
pp. \bfpage{343}--\blpage{86}.
\bpublisher{World Scientific},
\blocation{Singapore}
(\byear{2020})
\end{bchapter}
\endbibitem

\bibitem[\protect\citeauthoryear{Rosenthal}{1913}]{Rosenthal}
\begin{barticle}
\bauthor{\bsnm{Rosenthal}, \binits{A.}}:
\batitle{Beweis der {U}nm{ö}glichkeit ergodischer {G}assysteme}.
\bjtitle{Annalen der Physik}
\bvolume{347}(\bissue{14}),
\bfpage{796}--\blpage{806}
(\byear{1913})
\end{barticle}
\endbibitem

\bibitem[\protect\citeauthoryear{Plancherel}{1913}]{Plancherel}
\begin{barticle}
\bauthor{\bsnm{Plancherel}, \binits{M.}}:
\batitle{Beweis der {U}nm{ö}glichkeit ergodischer mechanischer {S}ysteme}.
\bjtitle{Annalen der Physik}
\bvolume{347}(\bissue{15}),
\bfpage{1061}--\blpage{63}
(\byear{1913})
\end{barticle}
\endbibitem

\bibitem[\protect\citeauthoryear{Sinai}{1970}]{Sinai}
\begin{barticle}
\bauthor{\bsnm{Sinai}, \binits{Y.G.}}:
\batitle{Dynamical systems with elastic reflections. ergodic properties of
  dispersing billiards}.
\bjtitle{Russian Mathematical Surveys}
\bvolume{25}(\bissue{2}),
\bfpage{137}--\blpage{89}
(\byear{1970})
\end{barticle}
\endbibitem

\bibitem[\protect\citeauthoryear{Sim{á}nyi}{2015}]{Simanyi}
\begin{barticle}
\bauthor{\bsnm{Sim{á}nyi}, \binits{N.}}:
\batitle{Further developments of {S}inai's ideas: The {B}oltzmann-{S}inai
  hypothesis}.
\bjtitle{To appear in the Abel Prize volume dedicated to Ya. G. Sinai. Arxiv
  version}
(\byear{2015})
\doiurl{10.48550/arXiv.1512.08272}
\end{barticle}
\endbibitem

\bibitem[\protect\citeauthoryear{Schiffrin and Wald}{2012}]{Schiffrin}
\begin{botherref}
\oauthor{\bsnm{Schiffrin}, \binits{J.S.}},
\oauthor{\bsnm{Wald}, \binits{R.M.}}:
Measure and probability in cosmology.
Physical Review D
\textbf{86}(023521)
(2012)
\end{botherref}
\endbibitem

\bibitem[\protect\citeauthoryear{Goldstein}{2012}]{Goldstein2012}
\begin{bchapter}
\bauthor{\bsnm{Goldstein}, \binits{S.}}:
\bctitle{Typicality and notions of probability in physics}.
In: \beditor{\bsnm{Ben-Menahem}, \binits{Y.}},
\beditor{\bsnm{Hemmo}, \binits{M.}} (eds.)
\bbtitle{Probability in Physics},
pp. \bfpage{59}--\blpage{71}.
\bpublisher{Springer},
\blocation{Berlin, Heidelberg}
(\byear{2012})
\end{bchapter}
\endbibitem

\bibitem[\protect\citeauthoryear{Lazarovici and
  Reichert}{2015}]{Lazarovici2015}
\begin{barticle}
\bauthor{\bsnm{Lazarovici}, \binits{D.}},
\bauthor{\bsnm{Reichert}, \binits{P.}}:
\batitle{Typicality, irreversibility and the status of macroscopic laws}.
\bjtitle{Erkenntnis}
\bvolume{80}(\bissue{4}),
\bfpage{689}--\blpage{716}
(\byear{2015})
\end{barticle}
\endbibitem

\bibitem[\protect\citeauthoryear{Wilhelm}{2019}]{Wilhelm}
\begin{botherref}
\oauthor{\bsnm{Wilhelm}, \binits{I.}}:
Typical: A theory of typicality and typicality explanations.
The British Journal for the Philosophy of Science
\textbf{73}(2)
(2019)
\end{botherref}
\endbibitem

\bibitem[\protect\citeauthoryear{Lanford}{1973}]{Lanford}
\begin{bchapter}
\bauthor{\bsnm{Lanford}, \binits{O.E.}}:
\bctitle{Entropy and equilibrium states in classical statistical mechanics}.
In: \beditor{\bsnm{Lenard}, \binits{A.}} (ed.)
\bbtitle{Statistical Mechanics and Mathematical Problems. Lecture Notes in
  Physics 20},
pp. \bfpage{1}--\blpage{113}
(\byear{1973})
\end{bchapter}
\endbibitem

\bibitem[\protect\citeauthoryear{Penrose}{1989}]{Penrose1989}
\begin{bbook}
\bauthor{\bsnm{Penrose}, \binits{R.}}:
\bbtitle{The Emperor's New Mind}.
\bpublisher{Oxford University Press},
\blocation{Oxford}
(\byear{1989})
\end{bbook}
\endbibitem

\bibitem[\protect\citeauthoryear{Penrose}{2004}]{Penrose2004}
\begin{bbook}
\bauthor{\bsnm{Penrose}, \binits{R.}}:
\bbtitle{The Road to Reality}.
\bpublisher{Vintage},
\blocation{London}
(\byear{2004})
\end{bbook}
\endbibitem

\bibitem[\protect\citeauthoryear{Lebowitz}{1993a}]{Lebowitz1993a}
\begin{barticle}
\bauthor{\bsnm{Lebowitz}, \binits{J.}}:
\batitle{Macroscopic laws, microscopic dynamics, time's arrow and boltzmann's
  entropy}.
\bjtitle{Physica A}
\bvolume{194},
\bfpage{1}--\blpage{24}
(\byear{1993})
\end{barticle}
\endbibitem

\bibitem[\protect\citeauthoryear{Lebowitz}{1993b}]{Lebowitz1993b}
\begin{botherref}
\oauthor{\bsnm{Lebowitz}, \binits{J.}}:
Boltzmann's entropy and time's arrow.
Physics Today,
32--38
(1993)
\end{botherref}
\endbibitem

\bibitem[\protect\citeauthoryear{Goldstein}{2001}]{Goldstein2001}
\begin{bchapter}
\bauthor{\bsnm{Goldstein}, \binits{S.}}:
\bctitle{Boltzmann's approach to statistical mechanics}.
In: \beditor{\bsnm{Bricmont}, \binits{J.}},
\beditor{\bsnm{D{ü}rr}, \binits{D.}},
\beditor{\bsnm{Galavotti}, \binits{M.C.}},
\beditor{\bsnm{Ghirardi}, \binits{G.}},
\beditor{\bsnm{Petruccione}, \binits{F.}},
\beditor{\bsnm{Zangh{ì}}, \binits{N.}} (eds.)
\bbtitle{Chance in Physics. Foundations and Perspectives},
pp. \bfpage{39}--\blpage{54}.
\bpublisher{Springer},
\blocation{Berlin}
(\byear{2001})
\end{bchapter}
\endbibitem

\bibitem[\protect\citeauthoryear{Reichert}{2020}]{Reichert}
\begin{botherref}
\oauthor{\bsnm{Reichert}, \binits{P.}}:
Essentially ergodic behaviour.
The British Journal for the Philosophy of Science
\textbf{74}(1)
(2020)
\end{botherref}
\endbibitem

\bibitem[\protect\citeauthoryear{Frigg and Werndl}{2015a}]{Frigg2015a}
\begin{barticle}
\bauthor{\bsnm{Frigg}, \binits{R.}},
\bauthor{\bsnm{Werndl}, \binits{C.}}:
\batitle{Reconceptualising equilibrium in {B}oltzmannian statistical mechanics
  and characterising its existence}.
\bjtitle{Studies in History and Philosophy of Modern Physics}
\bvolume{49}(\bissue{1}),
\bfpage{19}--\blpage{31}
(\byear{2015})
\end{barticle}
\endbibitem

\bibitem[\protect\citeauthoryear{Frigg and Werndl}{2015b}]{Frigg2015b}
\begin{barticle}
\bauthor{\bsnm{Frigg}, \binits{R.}},
\bauthor{\bsnm{Werndl}, \binits{C.}}:
\batitle{Rethinking {B}oltzmannian equilibrium}.
\bjtitle{Philosophy of Science}
\bvolume{82}(\bissue{5}),
\bfpage{1224}--\blpage{35}
(\byear{2015})
\end{barticle}
\endbibitem

\bibitem[\protect\citeauthoryear{Callen}{1960}]{Callen}
\begin{bbook}
\bauthor{\bsnm{Callen}, \binits{H.B.}}:
\bbtitle{Thermodynamics and an Introduction to Thermostatics}.
\bpublisher{Wiley},
\blocation{New York}
(\byear{1960})
\end{bbook}
\endbibitem

\bibitem[\protect\citeauthoryear{Reiss}{1996}]{Reiss}
\begin{bbook}
\bauthor{\bsnm{Reiss}, \binits{H.}}:
\bbtitle{Methods of Thermodynamics}.
\bpublisher{Dover},
\blocation{Minneaola, New York}
(\byear{1996})
\end{bbook}
\endbibitem

\bibitem[\protect\citeauthoryear{Lavis}{2005}]{Lavis2005}
\begin{barticle}
\bauthor{\bsnm{Lavis}, \binits{D.}}:
\batitle{Boltzmann and {G}ibbs: An attempted reconciliation}.
\bjtitle{Studies in History and Philosophy of Modern Physics}
\bvolume{36},
\bfpage{145}--\blpage{73}
(\byear{2005})
\end{barticle}
\endbibitem

\bibitem[\protect\citeauthoryear{Lavis}{2008}]{Lavis2008}
\begin{barticle}
\bauthor{\bsnm{Lavis}, \binits{D.}}:
\batitle{Boltzmann, {G}ibbs, and the concept of equilibrium}.
\bjtitle{Philosophy of Science}
\bvolume{75}(\bissue{5}),
\bfpage{682}--\blpage{96}
(\byear{2008})
\end{barticle}
\endbibitem

\bibitem[\protect\citeauthoryear{Frigg and Werndl}{2012}]{Frigg2012b}
\begin{barticle}
\bauthor{\bsnm{Frigg}, \binits{R.}},
\bauthor{\bsnm{Werndl}, \binits{C.}}:
\batitle{Demystifying typicality}.
\bjtitle{Philosophy of Science}
\bvolume{79},
\bfpage{917}--\blpage{29}
(\byear{2012})
\end{barticle}
\endbibitem

\bibitem[\protect\citeauthoryear{Vranas}{1998}]{Vranas}
\begin{barticle}
\bauthor{\bsnm{Vranas}, \binits{P.}}:
\batitle{Epsilon-ergodicity and the success of equilibrium statistical
  mechanics}.
\bjtitle{Philosophy of Science}
\bvolume{65}(\bissue{4}),
\bfpage{688}--\blpage{708}
(\byear{1998})
\end{barticle}
\endbibitem

\bibitem[\protect\citeauthoryear{Frigg and Werndl}{2011}]{Frigg2011}
\begin{barticle}
\bauthor{\bsnm{Frigg}, \binits{R.}},
\bauthor{\bsnm{Werndl}, \binits{C.}}:
\batitle{Explaining thermodynamic-like behavior in terms of epsilon-
  ergodicity}.
\bjtitle{Philosophy of Science}
\bvolume{78},
\bfpage{628}--\blpage{52}
(\byear{2011})
\end{barticle}
\endbibitem

\bibitem[\protect\citeauthoryear{Frigg and Werndl}{2012}]{Frigg2012a}
\begin{bchapter}
\bauthor{\bsnm{Frigg}, \binits{R.}},
\bauthor{\bsnm{Werndl}, \binits{C.}}:
\bctitle{A new approach to the approach to equilibrium}.
In: \beditor{\bsnm{Ben-Menahem}, \binits{Y.}},
\beditor{\bsnm{Hemmo}, \binits{M.}} (eds.)
\bbtitle{Probability in Physics},
pp. \bfpage{99}--\blpage{113}.
\bpublisher{Springer},
\blocation{Berlin}
(\byear{2012})
\end{bchapter}
\endbibitem

\bibitem[\protect\citeauthoryear{D{ü}rr and Teufel}{2009}]{DuerrTeufel}
\begin{bbook}
\bauthor{\bsnm{D{ü}rr}, \binits{D.}},
\bauthor{\bsnm{Teufel}, \binits{S.}}:
\bbtitle{Bohmian Mechanics. The Physics and Mathematics of Quantum Theory}.
\bpublisher{Springer},
\blocation{Berlin}
(\byear{2009})
\end{bbook}
\endbibitem

\end{thebibliography}

\end{document}